\documentclass[runningheads,a4paper]{llncs}

\usepackage{amssymb}
\usepackage{amsmath}
\usepackage{graphicx}

\usepackage{url}
\newtheorem{thm}{Theorem}
\newtheorem{lem}{Lemma}

\newcommand{\E}{\mathbb{E}}
\newcommand{\Prob}{\mathsf{P}}

\date{}

\begin{document}

\mainmatter

\title{Global clustering coefficient in scale-free \\ weighted and unweighted networks}

\titlerunning{Global clustering coefficient in scale-free networks}

\author{Liudmila~ Ostroumova Prokhorenkova\inst{1}\inst{2}}

\authorrunning{L. Ostroumova Prokhorenkova}

\institute{Yandex, Moscow, Russia
\and
Moscow State University, Moscow, Russia}

\maketitle

\begin{abstract}
In this paper, we present a detailed analysis of the global clustering coefficient in scale-free graphs.
Many observed real-world networks of diverse nature have a power-law degree distribution.
Moreover, the observed degree distribution usually has an infinite variance.
Therefore, we are especially interested in such degree distributions.
In addition, we analyze the clustering coefficient for both weighted and unweighted graphs.

There are two well-known definitions of the clustering coefficient of a graph: the global and the average local clustering coefficients.
There are several models proposed in the literature for which the average local clustering coefficient tends to a positive constant as a graph grows.
On the other hand, there are no models of scale-free networks with an infinite variance of the degree distribution and with an asymptotically constant global clustering coefficient.
Models with constant global clustering and finite variance were also proposed.
Therefore, in this paper we focus only on the most interesting case: we analyze the global clustering coefficient for graphs with an infinite variance of the degree distribution.

For unweighted graphs, we prove that the global clustering coefficient tends to zero with high probability and we also estimate the largest possible clustering coefficient for such graphs.
On the contrary, for weighted graphs, the constant global clustering coefficient can be obtained even for the case of an infinite variance of the degree distribution.

\end{abstract}

\section{Introduction}\label{sec:intro}

In this paper, we analyze the global clustering coefficient of graphs with a power-law degree distribution.
Namely, we consider a sequence of graphs with degree distributions following a regularly varying distribution $F$.
It was previously shown in \cite{GlobalClust} that if a graph has a power-law degree distribution with an infinite variance, then the global clustering coefficient tends to zero with high probability.
Namely, an upper bound for the number of triangles is obtained in~\cite{GlobalClust}.
In addition, the constructing procedure which allows to obtain the sequence of graphs with a superlinear number of triangles is presented.
However, the number of triangles in the constructed graphs grows slower than the upper bound obtained.
In this paper, we close this gap by improving the upper bound obtained in \cite{GlobalClust}.
Moreover, we also analyze graphs with multiple edges and show that weighted scale-free graphs with asymptotically constant global clustering coefficient and with an infinite variance of the degree distribution do exist.

The rest of the paper is organized as follows.
In the next section, we discuss several definitions of the clustering coefficient for weighted and unweighted graphs.
Then, in Section~\ref{sec:graph}, we formally define our restriction on a sequence of graphs.
In Sections~\ref{sec:unweighted} and~\ref{sec:weighted}, we analyze the global clustering coefficient for the unweighted and the weighted case respectively.
Section~\ref{sec:conclusion} concludes the paper.

\section{Clustering coefficients}\label{sec:clustering}

%Informally speaking, a clustering coefficient is supposed to measure difference between graph under consideration and a complete graph.

There are two well-known definitions of the clustering coefficient \cite{Math_Results,Newman} of an unweighted graph. The \emph{global clustering coefficient} $C_1(G_n)$ is the ratio of three times the number of triangles to the number of pairs of adjacent edges in $G_n$. The \emph{average local clustering coefficient} is defined as follows: $C_2(G_n) = \frac{1}{n} \sum_{i=1}^n C(i)$, where $C(i)$ is the local clustering coefficient for a vertex $i$: $C(i) = \frac{T^i}{P_2^i}$, where $T^i$ is the number of edges between the neighbors of the vertex $i$ and $P_2^i$ is the number of pairs of neighbors.
Note that both clustering coefficients equal $1$ for a complete graph.

It was mentioned in \cite{Math_Results,Newman} that in research papers either the average local or the global clustering coefficients are considered, and it is not always clear which definition is used.
On the other hand, these two clustering coefficients differ: e.g.,
it was demonstrated in \cite{GPA} that for networks based on the idea of preferential attachment the difference between these two clustering coefficients is crucial.

It is also reasonable to study the global clustering coefficient for graphs with multiple edges. This agrees well with reality, for example, the Web host graph has a lot of multiple edges: there can be several edges between the pages of two hosts.
And even in the Internet graph (vertices are web pages and edges are links between them) multiple edges occur.

We refer to the paper \cite{MultClust} for the definition of the global clustering coefficient for weighted graphs.
They propose the following generalization of the global clustering coefficient to multigraphs:
$$
C_1(G) = \frac{\text{total value of closed triplets}}{\text{total value of triplets}}\,.
$$

There are several ways to define the value of a triplet.
First, the triplet value can be defined as the \textit{arithmetic mean} of
the weights of the ties that make up the triplet.
Second, it can be defined as the \textit{geometric mean} of the weights of the ties.
Third, it can be defined as the \textit{maximum or minimum value} of the weights of the ties.
In addition to these methods proposed in \cite{MultClust}, we also propose the following natural definition of the weight: the weight of a triplet is the \textit{product} of the weights of the ties. This definition agrees with the following property: the total value of all triplets located in a vertex is close to its degree squared.

%This definition allows one property to be satisfied: the total value of triplets centered on a vertex $v$ is equal to $\frac {deg(v) (deg(v) - 1)} {2}$.

\section{Scale-free graphs}\label{sec:graph}

We consider a sequence of graphs $\{G_n\}$.
Each graph $G_n$ has $n$ vertices.
As in~\cite{GlobalClust}, we assume that the degrees of the vertices are independent random variables following a {\it regularly varying} distribution with a cumulative distribution function $F$ satisfying
\begin{equation}\label{eq:regular}
1-F(x)=L(x)x^{-\gamma},\quad x>0,
\end{equation}
where $L(\cdot)$ is a slowly varying function, that is, for any fixed constant $t>0$
\[\lim_{x\to\infty}\frac{L(tx)}{L(x)}=1.\]
There is another obvious restriction on the function $L(\cdot)$: the function $1-L(x)x^{-\gamma}$ must be a cumulative distribution function of a random variable taking positive integer values with probability 1.

Note that Equation~\eqref{eq:regular} describes a broad class of heavy-tailed distributions without imposing the rigid Pareto assumption.
The power-law distribution with parameter $\gamma + 1$ corresponds to the cumulative distribution $1-F(x)=L(x) x^{-\gamma}$.
Further by $\xi, \xi_1, \xi_2, \ldots$ we denote random variables with the distribution $F$.
Note that for any $\alpha < \gamma$ the moment $\E \xi^{\alpha}$ is finite.

%\textbf{TODO}

%Note that if the degrees are independent random variables with a distribution whose cumulate $F(x)$ has the form $1-F(x)=L(x) x^{-\gamma}$,
%then this generates a graph with power law degree distribution $\gamma + 1$, and thus for any $\alpha < \gamma$ the moments $E X^{\alpha}$ are finite.

Models with $\gamma > 2$ and with the global clustering coefficient tending to some positive constant were already proposed (see, e.g., \cite{GPA}).
Therefore, in this paper we consider only the case $1<\gamma < 2$.

One small problem remains: we can construct a graph with a given degree distribution only if the sum of all degrees is even. This problem is easy to solve: we can either regenerate
the degrees until their sum is even or we can add 1 to the last variable if their sum is odd \cite{Configuration}.
For the sake of simplicity we choose the second option, i.e., if $\sum_{i=1}^n \xi_i$ is odd, then we replace $\xi_n$ by $\xi_n+1$.
It is easy to see that this modification does not change any of our results, therefore,
further we do not focus on the evenness.

\section{Auxiliary results}

In this section, we prove several auxiliary lemmas.
These lemmas generalize several results from \cite{GlobalClust}.
In order to prove these lemmas we use the following theorem (see, e.g., \cite{RegVar}).

\begin{thm}[Karamata's theorem]\label{thm:karamata}
Let $L$ be slowly varying and locally bounded in $[ x_0,\infty ]$ for some $x_0 \geq 0$. Then

\begin{enumerate}
\item for $\alpha > -1$
    $$
        \int_{x_0}^x t^{\alpha} L(t) dt  = (1+o(1))  (\alpha +1)^{-1} x^{\alpha+1} L(x),\,\,\,\, x\to\infty\,.
    $$
\item for $\alpha < -1$
    $$
        \int_{x}^{\infty} t^{\alpha} L(t) dt   = - (1+o(1)) (\alpha +1)^{-1} x^{\alpha+1} L(x),\,\,\,\, x\to\infty\,.
    $$
\end{enumerate}
\end{thm}

We also use the following known lemma (its proof can be found, e.g., in~\cite{fresh_model}).

\begin{lem}\label{lem:moments}
Let $\xi_1, \dots, \xi_n$ be mutually independent random variables, $\E \xi_i = 0$, $\E |\xi_i|^{\alpha} < \infty$, $1\le \alpha \le 2$, then
$$
\E \left( |\xi_1 + \ldots + \xi_n|^{\alpha} \right) \le 2^{\alpha} \left(\E\left(|\xi_1|^{\alpha}\right)+\ldots+\E\left(|\xi_n|^{\alpha}\right) \right)\,.
$$
\end{lem}

We need the following notation:
$$
S_{n,c}(x) = \sum_{i=1}^n \xi_i^c I\left[\xi_i>x\right]\,,
$$
$$
\bar S_{n,c}(x) = \sum_{i=1}^n \xi_i^c I\left[\xi_i\le x\right]\,,
$$
here $c,x \ge 0$.

\begin{lem}\label{lem:1}
Fix any $c$ such that $0\le c<\gamma$, any $\beta$ such that $1 < \beta < \gamma/c$ and $\beta \le 2$, and any $\varepsilon > 0$.
Then for any $x = x(n)>0$ such that $x(n) \to \infty$ we have
$$
\E S_{n,c}(x) = \frac{\gamma}{\gamma-c}n\,x^{c - \gamma} L\left(x\right)(1+o(1)),\,\,\, n \to \infty\,,
$$
$$
\Prob \left( |S_{n,c}(x) - \E S_{n,c}(x)| > \varepsilon \, \E S_{n,c}(x) \right) = O \left( \left( \frac {x^\gamma}{n \, L\left(x\right)} \right)^{\beta-1} \right)\,.
$$
\end{lem}

\begin{proof}

We now assume that $c>0$.

First, we estimate the expectation of $S_{n,c}(x)$:
\begin{multline*}
\E S_{n,c}(x)
= n \int_{x}^{\infty} t^c d F(t)
= -n \int_{x}^{\infty} t^c \, d (1- F(t))
\\
= - n \, t^c (1- F(t)) \bigg|_{x}^{\infty} + n \, c \int_{x}^{\infty} t^{c-1} (1-F(t)) \, d t
\\
= n \, x^{c - \gamma}  L\left(x\right) + n \, c \int_{x}^{\infty} t^{c - 1 -\gamma}L(t) \, d t
\\
\sim n \, x^{c - \gamma} L\left(x\right)
- n \,c  (c -\gamma)^{-1} x^{c-\gamma}  L\left(x\right)
= \frac{\gamma}{\gamma-c}n \, x^{c - \gamma} L\left(x\right)\,.
\end{multline*}
Then, we estimate
$$
\E \left(\xi^c I\left[\xi > x\right]\right)^\beta
= \frac{1}{n} \E S_{n,c\beta}(x) \sim
%= - \int_{x}^{\infty} t^{\beta c} \, d (1- F(t))
%\\
%= - t^{\beta c } (1- F(t)) \bigg|_{x}^{\infty}
%+ \int_{x}^{\infty} (1-F(t)) \, d t^{\beta c}
%\\
%= x^{\beta c} x^{-\gamma} L\left(x\right) + \beta c \int_{x}^{\infty} t^{\beta c -\gamma-1}L(t) \, d t
%\\
%\sim x^{\beta c -\gamma} L\left(x\right) + \beta c(\gamma-\beta c)^{-1} x^{\beta c -\gamma}  L\left(x\right)
\frac{\gamma}{\gamma -\beta c}x^{\beta c -\gamma} L\left(x\right)
$$
and get
\begin{multline*}
\Prob \left( |S_{n,c}(x) - \E S_{n,c}(x)| > \varepsilon \, \E S_{n,c}(x) \right) \le
\frac{\E |S_{n,c}(x) - \E S_{n,c}(x)|^\beta}{(\varepsilon \, \E S_{n,c}(x))^\beta}
\\
= O\left( \frac{ n \E \left(\xi^c I\left[\xi > x\right]\right)^\beta}{(\E S_{n,c}(x))^\beta} \right)
= O \left( \frac{n \, x^{\beta c-\gamma} L\left(x\right)}{n^\beta x^{\beta(c - \gamma)} \left(L\left(x\right)\right)^\beta} \right)
\\
= O \left( n^{1-\beta}x^{-(1-\beta)\gamma} \left(L\left(x\right)\right)^{1-\beta} \right)\,.
\end{multline*}

The case $c=0$ can be considered similarly:
$$
\E S_{n,0}(x)
= n \Prob(\xi > x)
= n \, x^{- \gamma} L\left(x\right)\,,
$$
%In this case $\gamma/c > 2$, so we can take $\beta = 2$.
\begin{multline*}
\Prob \left( |S_{n,0}(x) - \E S_{n,0}(x)| > \varepsilon \, \E S_{n,0}(x) \right)
%\frac{\mathrm{Var} \, S_{n,0}(x)}{(\varepsilon \, \E S_{n,0}(x))^2}
%= O \left(\frac{1} {n x^{-\gamma} L\left(x\right)} \right)\,.
= O\left( \frac{ n x^{- \gamma} L\left(x\right)}{(n \, x^{- \gamma} L\left(x\right))^\beta} \right) \\
= O \left( \left( \frac {x^\gamma}{n \, L\left(x\right)} \right)^{\beta-1} \right)\,.
\end{multline*}

\end{proof}

\begin{lem}\label{lem:2}
Fix any $c$ such that $c>\gamma$
%, any $\beta$ such that $1 < \beta \le 2$,
and any $\varepsilon > 0$.
Then for any $x = x(n)>0$ such that $x(n) \to \infty$ we have
$$
\E \bar S_{n,c}(x) = \frac{\gamma}{c - \gamma}n \, x^{c - \gamma} L\left(x\right)(1+o(1)), \,\,\, n\to \infty\,,
$$
$$
\Prob \left( |\bar S_{n,c}(x) - \E \bar S_{n,c}(x)| > \varepsilon \, \E \bar S_{n,c}(x) \right) =
O \left( \frac {x^\gamma}{n \, L\left(x\right)} \right)\,.
$$
\end{lem}

\begin{proof}

Again, first we estimate the expectation of $\bar S_{n,c}(x)$:
\begin{multline*}
\E \bar S_{n,c}(x)
= n \int_{0}^{x} t^c d F(t)
= -n \int_0^{x} t^c \, d (1- F(t))
\\
= - n \, t^c (1- F(t)) \bigg|_0^{x} + n \, c \int_0^{x} t^{c-1} (1-F(t)) \, d t
\\
= - n \, x^{c - \gamma}  L\left(x\right) + n \, c \int_0^{x} t^{c - 1 -\gamma}L(t) \, d t
\\
\sim - n \, x^{c - \gamma}  L\left(x\right)
+ n \,c  (c -\gamma)^{-1} x^{c-\gamma}  L\left(x\right)
= \frac{\gamma}{c - \gamma}n \, x^{c - \gamma} L\left(x\right)\,.
\end{multline*}
Then, we estimate
$$
\E \left(\xi^c I\left[\xi \le x\right]\right)^2
= \frac{1}{n} \bar S_{n,2c}(x)
%= - \int_0^{x} t^{2 c} \, d (1- F(t))
%\\
%= - t^{2 c } (1- F(t)) \bigg|_0^{x}
%+ \int_0^{x} (1-F(t)) \, d t^{2 c}
%\\
%= x^{2 c - \gamma} L\left(x\right) + 2 c \int_0^{x} t^{2 c -\gamma-1}L(t) \, d t
%\\
\sim %x^{2 c -\gamma} L\left(x\right) + 2 c (2 c - \gamma)^{-1} x^{2 c -\gamma}  L\left(x\right)
%=
\frac{\gamma}{2 c - \gamma}x^{2 c -\gamma} L\left(x\right)\,
$$
and get
\begin{multline*}
\Prob \left( |\bar S_{n,c}(x) - \E \bar S_{n,c}(x)| > \varepsilon \, \E \bar S_{n,c}(x) \right) \le
\frac{\E |\bar S_{n,c}(x) - \E \bar S_{n,c}(x)|^2}{( \varepsilon \, \E \bar S_{n,c}(x))^2}
\\
= O\left( \frac{ n \E \left(\xi^c I\left[\xi \le x \right]\right)^2}{(\E \bar S_{n,c}(x))^2} \right)
= O \left( \frac{n\, x^{2 c-\gamma} L\left(x\right)}{n^{2}x^{2(c - \gamma)} \left(L\left(x\right)\right)^2} \right)
= O \left( \frac {x^\gamma}{n \, L\left(x\right)} \right)\,.
\end{multline*}

\end{proof}

We prove two more lemmas.
Put $\xi_{max} = \max\{\xi_1, \ldots, \xi_n\}$.

\begin{lem}\label{lem:ximax}
For any $\varepsilon > 0$ and any $\alpha > 0$
$$
\Prob\left(\xi_{max} > n^{\frac{1}{\gamma}-\varepsilon}\right) = 1 - O\left(n^{-\alpha}\right)\,.
$$
Also, for any $\delta < \gamma\varepsilon$
$$
\Prob\left(\xi_{max} \le n^{\frac{1}{\gamma}+\varepsilon}\right) = 1 - O\left(n^{-\delta}\right)\,.
$$
\end{lem}

\begin{proof}

\begin{multline*}
\Prob\left(\xi_{max} \le n^{\frac{1}{\gamma}-\varepsilon}\right)
= \left[ \Prob\left(\xi \le n^{\frac{1}{\gamma}-\varepsilon}\right) \right]^n
= \exp\left(n \log \left(1 - \Prob(\xi > n^{\frac{1}{\gamma}-\varepsilon})\right)\right)
\\
= \exp\left(n \log \left(1 - L\left(n^{\frac{1}{\gamma}-\varepsilon}\right) n^{-\gamma\left(\frac{1}{\gamma}-\varepsilon\right)}\right)\right)
\\
= \exp\left(- n \, L\left(n^{\frac{1}{\gamma}-\varepsilon}\right) n^{-\gamma\left(\frac{1}{\gamma}-\varepsilon\right)} (1 + o(1))\right)
\\
= \exp\left(-L\left(n^{\frac{1}{\gamma}-\varepsilon}\right) n^{\gamma\varepsilon} (1 + o(1))\right) = O\left(n^{-\alpha}\right)\,,
\end{multline*}
\begin{multline*}
\Prob(\xi_{max} > n^{\frac{1}{\gamma}+\varepsilon})
\le n \, \Prob\left(\xi > n^{\frac{1}{\gamma}+\varepsilon}\right)
\le n \, L\left(n^{\frac{1}{\gamma} + \varepsilon}\right) n^{-\gamma\left(\frac{1}{\gamma}+\varepsilon \right)}
= O\left(n^{-\delta}\right).
\end{multline*}

\end{proof}

\begin{lem}\label{lem:p2upper}
For any $\varepsilon > 0$  and any $\delta < \frac{\gamma\varepsilon}{\gamma+2}$
$$
\Prob \left( \bar S_{n,2}(\infty) \le n^{\frac{2}{\gamma} + \varepsilon} \right) = 1 - O\left(n^{-\delta}\right)\,.
$$
\end{lem}

\begin{proof}

Choose $\varphi$ such that $\frac{\delta}{\gamma} < \varphi < \frac{\varepsilon}{\gamma+2}$.
From Lemma~\ref{lem:ximax} we get
$$
\Prob(\xi_{max} \le n^{\frac{1}{\gamma}+\varphi}) = 1 - O\left(n^{-\delta}\right)\,.
$$
From Lemma~\ref{lem:1} and Lemma~\ref{lem:2}, with probability
$$
1 - O \left(  \frac {n^{\gamma\left(\frac 1 \gamma - \varphi\right)}}{n \, L\left(n^{\frac 1 \gamma - \varphi}\right)}  \right) = 1 - O\left(n^{-\delta}\right)
$$
we have
$$
\bar S_{n,2}\left(n^{\frac 1 \gamma - \varphi}\right) \le (1+\varepsilon) \frac{\gamma}{2 - \gamma} n^{\frac 2 \gamma + \varphi \gamma - 2 \varphi} L\left(n^{\frac 1 \gamma - \varphi}\right)\,,
$$
$$
S_{n,0}\left(n^{\frac 1 \gamma - \varphi}\right) \le (1+\varepsilon) \, n^{\varphi \gamma} L\left(n^{\frac 1 \gamma - \varphi}\right)\,.
$$

In this case,
\begin{multline*}
\bar S_{n,2} \left(n^{\frac 1 \gamma + \varphi}\right)
\le \bar S_{n,2} \left(n^{\frac 1 \gamma - \varphi}\right) + \xi_{max} \, S_{n,0}\left(n^{\frac 1 \gamma - \varphi}\right)
\\
\le (1+\varepsilon) \frac{\gamma}{2 - \gamma} n^{\frac 2 \gamma + \varphi \gamma - 2 \varphi} L\left(n^{\frac 1 \gamma - \varphi}\right) +
n^{\frac{2}{\gamma}+2\varphi} (1+\varepsilon) \, n^{\varphi \gamma} L\left(n^{\frac 1 \gamma - \varphi}\right) \le
n^{\frac{2}{\gamma}+\varepsilon}
\end{multline*}
for large enough $n$.
This concludes the proof.

\end{proof}

Note that we estimated only the upper bound for $\bar S_{n,2}(\infty)$, since the lower bound can be obtained using the lower bound for $\xi_{max}$. Here we may use the inequality $S_{n,2}(\infty) \ge \xi_{max}^2$.

\section{Clustering in unweighted graphs}\label{sec:unweighted}

\subsection{Previous results}

The behavior of the global clustering coefficient in scale-free unweighted graphs was considered in \cite{GlobalClust}.
In the case of an infinite variance, the reasonable question is whether there exists a simple graph (i.e., a graph without loops and multiple edges) with a given degree distribution.
The following theorem is proved in \cite{GlobalClust}.

\begin{thm}\label{thm:existence}
%For any $\alpha$ such that $1< \alpha < \gamma$ with probability $1-O\left(n^{1-\alpha}\right)$
With hight probability
there exists a simple graph on $n$ vertices with the degree distribution defined in Section~\ref{sec:graph}.
\end{thm}

So, with high probability such a graph exists and it is reasonable to discuss its global clustering coefficient.
The following upper bound on the global clustering coefficient is obtained in \cite{GlobalClust}.

\begin{thm}\label{thm:cluster}
%For any $\varepsilon>0$ and any $\alpha$ such that $1<\alpha<\gamma$ with probability $1-O(n^{1-\alpha})$
For any $\varepsilon > 0$ with high probability the global clustering coefficient satisfies the following inequality
$$
C_1(G_n) \le n^{ - \frac{\left(\gamma - 2\right)^2}{2\gamma}+\varepsilon}\,.
$$
\end{thm}
Taking small enough $\varepsilon$ one can see that with high probability $C_1(G_n) \to 0$ as $n$ grows.

In addition, using simulations and empirical observations, the authors of \cite{GlobalClust} claimed that with high probability there exists a graph with $\sim n^{\frac{3}{\gamma+1}}$ triangles and with the required degree distribution, while the theoretical upper bound on the number of triangles is $n^{2 - \frac \gamma 2}$.
For the considered case $1 < \gamma < 2$ we have $\frac{3}{\gamma+1} < 2 - \frac \gamma 2$ and there is a gap between the number of constructed triangles and the obtained upper bound.

Further in this section we close this gap by improving the upper bound.
We also rigorously prove the lower bound.

\subsection{Upper bound}

We prove the following theorem.

\begin{thm}\label{thm:upperbound}
For any $\varepsilon>0$ and any $\alpha$ such that $0<\alpha<{\frac{1}{\gamma+1}}$ with probability $1-O(n^{-\alpha})$ the global clustering coefficient satisfies the following inequality
$$
C_1(G_n) \le n^{ - \frac{\left(2 - \gamma \right)}{\gamma(\gamma+1)}+\varepsilon}\,.
$$
\end{thm}

\begin{proof}

The global clustering coefficient is
$$
C_1(G_n) = \frac{3 \cdot T(n) }{P_2(n)},
$$
where $T(n)$ is the number of triangles and $P_2(n)$ is the number of pairs of adjacent edges in $G_n$.

Since $P_2(n) \ge \xi_{max}(\xi_{max} - 1)/2$. Therefore, from Lemma~\ref{lem:ximax} we get that for any $\delta > 0$ with probability $1 - O\left(n^{-\alpha}\right)$
$$
P_2(n) > n^{\frac{2}{\gamma}-\delta}\,.
$$

It remains to estimate $T(n)$. Obviously, for any $x$
\begin{equation}\label{eq:T(n)}
T(n) \le \left|\{i: \xi_i > x\}\right|^3  + \sum_{i: \xi_i \le x} \xi_i^2\,.
\end{equation}
The first term in \eqref{eq:T(n)} is the upper bound for the number of triangles with all vertices among the set $\{i: \xi_i > x\}$. The second term is the upper bound for the number of triangles with at least one vertex among $\{i: \xi_i \le x\}$.

From Lemma~\ref{lem:1} and Lemma~\ref{lem:2} we get
$$
\left|\{i: \xi_i > x\}\right| = S_{n,0}(x) \le (1+\varepsilon) n\,x^{ - \gamma} L\left(x\right) \,,
$$
$$
\sum_{i: \xi_i \le x} \xi_i^2 = \bar S_{n,2}(x) \le (1+\varepsilon)  \frac{\gamma}{2 - \gamma}n \, x^{2 - \gamma} L\left(x\right)
$$
with probability $1 - O \left( \frac {x^\gamma}{n \, L\left(x\right)} \right)$\,.

Now we can fix $x = n^\frac{1}{\gamma+1}$. So, with probability
$$
1 - O \left( \frac {n^{-\frac{1}{\gamma+1}}}{L\left(n^\frac{1}{\gamma+1}\right)} \right) = 1-O(n^{-\alpha})
$$
we have
$$
T(n) \le n^{\frac{3}{\gamma + 1}+\delta}
$$

%Form \textbf{[TODO: lemmas]} we get that with probability $1 - O\left(n^{1-\alpha}\right)$
%$$
%\left|\{i: \xi_i \ge n^{\frac{1}{\gamma+1}}\}\right|^3 \le  n^{\frac{3}{\gamma+1}+\delta}\,,
%$$
%$$
%\sum_{i: \xi_i< n^{\frac{1}{\gamma+1}}} \xi_i^2 \le  n^{\frac{3}{\gamma+1} + \delta}\,.
%$$
%From \eqref{eq:T(n)} we get
%$$
%T(n) \le 2n^{\frac{3}{\gamma+1} + \delta}\,.
%$$
Taking small enough $\delta$, we obtain
$$
C_1(G_n) \le n^{\varepsilon-\frac{2 - \gamma}{\gamma(\gamma+1)}}\,.
$$
This concludes the proof.

\end{proof}

\subsection{Lower bound}

%We need the following lemma on the number of edges in the graph (see \cite{GlobalClust}).
%\begin{lem}\label{lem:edges}
%For any $\theta$ such that $1<\theta<\gamma$ with probability $1-O(n^{1-\theta})$ the number of edges $E(G_n)$ in our graph satisfies the following inequalities:
%$$
%\frac{n \E \xi}{4} \le E(G_n) \le \frac{3n\E\xi}{4}\,.
%$$
%\end{lem}

We prove the following theorem.

\begin{thm}\label{thm:lowerbound}
For any $\varepsilon>0$ and any $\alpha$ such that $0<\alpha<min\{\frac{\gamma\varepsilon}{\gamma+2},\frac{1}{\gamma+1},\gamma-1\}$ with probability $1-O(n^{-\alpha})$ there exists a graph with the required degree distribution and the global clustering coefficient satisfying the following inequality
$$
C_1(G_n) \ge n^{- \frac{\left(2 - \gamma \right)}{\gamma(\gamma+1)} - \varepsilon}\,.
$$
\end{thm}

\begin{proof}

Again,
$$
C_1(G_n) = \frac{3 \cdot T(n) }{P_2(n)}\,.
$$

The upper bound for $P_2(n)$ follows from Lemma~\ref{lem:p2upper}.
Fix $\varepsilon'$ such that $\frac{\alpha(\gamma+2)}{\gamma}< \varepsilon' < \varepsilon$.
Then,
$$
\Prob \left( P_2(n) \le n^{\frac{2}{\gamma} + \varepsilon'} \right) \ge \Prob \left(\bar S_{n,2}(\infty) \le n^{\frac{2}{\gamma} + \varepsilon'} \right) = 1 - O\left(n^{-\alpha}\right)\,.
$$

\end{proof}

Now we present the lower bound for $T(n)$.
Fix any $\delta$ such that $0<\gamma\delta < \min \left\{ \frac{1}{\gamma+1} - \alpha, \frac{\varepsilon - \varepsilon'}{3}  \right\}$.
It follows from Lemma~\ref{lem:1} that with probability
$1 - O\left(n^{-\alpha}\right)$
$$
S_{n,0}\left(n^{\frac{1}{\gamma+1}+\delta}\right) \le (1+\varepsilon) \,  n^{ \frac{1}{\gamma+1}-\gamma\delta} L\left(n^{\frac{1}{\gamma+1}+\delta}\right) \le n^{\frac{1}{\gamma+1} + \delta} \,.
$$
Let us denote by $A$ the set of vertices whose degrees are greater than $n^{\frac{1}{\gamma+1}+\delta}$. The size of $A$ equals
 $S_{n,0}\left(n^{\frac{1}{\gamma+1}+\delta}\right)$.
Since the number of vertices in $A$ is not greater than the minimum degree in $A$, a clique on $A$ can be constructed.
Therefore, with probability $1 - O\left(n^{-\alpha}\right)$
$$
S_{n,0}\left(n^{\frac{1}{\gamma+1}+\delta}\right) \ge (1-\varepsilon) \,  n^{ \frac{1}{\gamma+1}-\gamma\delta} L\left(n^{\frac{1}{\gamma+1}+\delta}\right)
%\ge n^{\frac{1}{\gamma+1} - \frac{\varepsilon - \varepsilon'}{3}} \,,
$$
and
$$
3 T(n) \ge 3 {S_{n,0}\left(n^{\frac{1}{\gamma+1}+\delta}\right)\choose 3} \ge  n^{ \frac{3}{\gamma+1}-(\varepsilon - \varepsilon')}\,.
$$
Finally,  we get
$$
C_1(G_n) = \frac{3 \cdot T(n) }{P_2(n)} \ge \frac{n^{ \frac{3}{\gamma+1}-(\varepsilon - \varepsilon')}}{n^{\frac{2}{\gamma} + \varepsilon'}} = n^{- \frac{\left(2 - \gamma \right)}{\gamma(\gamma+1)} - \varepsilon} \,.
$$

It remains to prove that after we constructed a clique on the set $A$, with high probability we still can construct a graph without loops and multiple edges.
This can be easily proved similarly to Theorem~\ref{thm:existence}.
Namely, we use the following theorem by Erd\H{o}s and Gallai~\cite{ErdosGallai}.

\begin{thm}[Erd\H{o}s--Gallai]
A sequence of non-negative integers ${d_1 \geq \ldots \geq d_n}$ can be represented as the degree sequence of a finite simple graph on $n$ vertices if and only if
\begin{enumerate}
\item $d_1+\ldots+d_n$ is even;
\item $\sum^{k}_{i=1}d_i\leq k(k-1)+ \sum^n_{i=k+1} \min(d_i,k)$ holds for $1\leq k\leq n$.
\end{enumerate}
\end{thm}

Let us order the random variables $\xi_1, \ldots, \xi_n$ and obtain the ordered sequence $d_1 \ge \ldots \ge d_n$.
In order to apply the theorem of Erd\H{o}s and Gallai we assume that the set $A$ is now a single vertex with the degree
$$
deg(A)  = S_{n,1}\left(n^{\frac{1}{\gamma+1}+\delta}\right) - 2 { S_{n,0}\left(n^{\frac{1}{\gamma+1}+\delta}\right) \choose 2}
< S_{n,1}\left(n^{\frac{1}{\gamma+1}+\delta}\right) \,.
%- 2{S_{n,1}\left(n^{\frac{1}{\gamma+1}+\delta}\right) \choose 3}\,.
$$

It is sufficient to prove that with probability $1 - O(n^{-\alpha})$ the following condition
is satisfied
%\begin{equation}\label{eq:degcondition0}
%deg(A) \leq \sum^n_{i=|A|+1} 1 = n - |A|\,,
%\end{equation}
\begin{equation}\label{eq:degcondition}
deg(A) + \sum^{k}_{i=|A|+1}d_i\leq (k-|A|)(k-|A|+1)+ \sum^n_{i=k+1} \min(d_i,k-|A|+1)
\end{equation}
for all $k \ge |A|$.

Let us now prove that with probability $1 - O(n^{-\alpha})$ this condition is satisfied.
For some large enough $C$ if $k > C \sqrt n$, then
$$
deg(A) + \sum^{k}_{i=|A|+1}d_i\leq (k-|A|)(k-|A|+1)\,.
$$
This holds since with probability $1 - O(n^{-\alpha})$
$$
|A| = S_{n,0}\left(n^{\frac{1}{\gamma+1}+\delta}\right) \le (1+\varepsilon)n^{ \frac{1}{\gamma+1}-\delta} L\left(n^{\frac{1}{\gamma+1}+\delta}\right) \le n^{ \frac{1}{\gamma+1}}\,.
$$
and
the sum of all degrees grows linearly with $n$:
$$
\Prob \left(|S_{n,1}(0) - n\E \xi|  > \frac n 2\, \E \xi \right)
%\frac{2^{\alpha+1} \E \left| \sum_{i=1}^n \left( \xi_i - \E\xi\right)\right|^{\alpha+1}}{n^{\alpha+1} (\E\xi)^\theta}
\le \frac{ 4^{\alpha+1} n\, \E|\xi - \E\xi|^{\alpha+1}}{n^{{\alpha+1}} (\E\xi)^{\alpha+1}}
= O\left(n^{-\alpha}\right)\,.
$$
Here we used that $\alpha + 1 < \gamma$.

Finally, consider the case $k \le C \sqrt n$.
Note that $\min(d_i,k-|A|+1) > 1$, so
$$
\sum^n_{i=k+1} \min(d_i,k-|A|+1)   \ge n - C\sqrt{n}\,.
$$
It remains to show that with probability $1 - O(n^{-\alpha})$
$$
deg(A) + \sum^{k}_{i=|A|+1}d_i\leq n - C\sqrt{n}\,.
$$
It is sufficient to show that
$$
\sum_{i=1}^{[C \sqrt n]}d_i\leq n - C\sqrt{n}\,.
$$

%From this \eqref{eq:degcondition0} follows.
%and $|A|$ is of order $n^{ \frac{1}{\gamma+1}-\delta} L\left(n^{\frac{1}{\gamma+1}+\delta}\right)$
%And now we only have to prove that
%$$
%n^{ \frac{1}{\gamma+1}} + \sum^{[C \sqrt n]}_{i=1}d_i \leq n - C\sqrt{n} \,.
%$$
This inequality is easy to prove using Lemma~\ref{lem:1}.
For any $\frac{1}{3\gamma} < \delta < \frac{1}{2\gamma}$ with probability $1- O\left(n^{-1/2}\right)$ we have
$$
S_{n,0}\left(n^\delta\right) > C \sqrt n
$$
and
$$
S_{n,1}\left(n^\delta\right) \le n^{\frac{2\gamma+1}{3\gamma}}\leq n - C\sqrt{n}\,.
$$
Therefore, the condition~\eqref{eq:degcondition} is satisfied.

\section{Clustering in weighted graphs}\label{sec:weighted}

In this section, we analyze the global clustering coefficient of graphs with multiple edges.
First, let us note that the case when we allow both loops and multiple edges is not very interesting: we can get a high clustering coefficient just by avoiding triplets.
Namely, we can construct several triangles and then just create loops in all vertices.
Then, we can connect the remaining half-edges for the vertices with odd degrees.
Therefore, further we assume that loops are not allowed.
We show that even with this restriction it is possible to obtain a constant global clustering coefficient.

Several definitions of the global clustering coefficient for graphs with multiple edges are presented in Section~\ref{sec:clustering}.
The following theorem holds for any definition of the global clustering coefficient $C_1(G_n)$.

\begin{thm}
Fix any $\delta > 0$. For any $\alpha$ such that $0 < \alpha < \frac{\gamma-1}{\gamma+1}$ with probability $1 - O\left(n^{-\alpha}\right)$
there exists
a multigraph with the required degree distribution and the global clustering coefficient satisfying the following inequality
$$
C_1(G_n) \ge \frac{2 - \gamma}{2+\gamma} -\delta\,.
$$
\end{thm}

\begin{proof}

Fix some $\varepsilon > 0$. From Lemma~\ref{lem:1} with $c = 0$ it follows that  with probability
$1 - O\left(\frac{x^\gamma}{n \, L(x)}\right)$
\begin{equation}\label{eq:x}
(1-\varepsilon) \, n \, x^{- \gamma} L\left(x\right) \le S_{n,0}(x) \le (1+\varepsilon) \, n \, x^{- \gamma} L\left(x\right)\,.
\end{equation}

Let us prove that for large enough $n$ there always exists such $x_0$ that
\begin{equation}\label{eq:x0}
(1+\varepsilon) \, n \, x_0^{- \gamma} L\left(x_0\right) \le x_0 \le (1+2\varepsilon)\, n \, x_0^{- \gamma} L\left(x_0\right)\,.
\end{equation}
In other words, we want to find such $x_0$ that
%$$
%(1+\varepsilon) \, n \le \frac {x}{x^{- \gamma} L\left(x\right)}  \le (1+2\varepsilon)\, n \,.
%$$
%Obviously, the function $f(x) = \frac {x}{x^{- \gamma} L\left(x\right)}$ grows to infinity.
$$
\frac{1}{(1+2\varepsilon)\,n} \le \frac{x_0^{ - \gamma} L\left(x_0\right)}{x_0} \le \frac {1}{(1+\varepsilon)\,n} \,.
$$
Recall that $x^{- \gamma} L\left(x\right) = 1 - F(x)$, where $F(x)$ is a cumulative distribution function.
Therefore, $f(x) = \frac{x^{ - \gamma} L\left(x\right)}{x}$ monotonically decreases to zero on $(0,\infty)$.
The only problem is that $f(x)$ is a discontinuous function.
In order to guarantee the existence of the required value $x_0$, we have to prove that (for large enough $n$) if $f(x) < \frac 1 n$, then $f(x+) - f(x-) < \frac{\varepsilon}{(1+\varepsilon)(1+2\varepsilon)n}$.
This can be proved as follows.
For the function $F(x)$ it is obvious that if $1 - F(x) < \frac 1 n$, then $|F(x+) - F(x-)| < \frac{1}{n}$.
Therefore, in this case, $|f(x+) - f(x-)| < \frac{1}{x \, n}$.
For large enough $n$ (and this leads to large enough $x$) we have $\frac{1}{x \, n} \le \frac{\varepsilon}{(1+\varepsilon)(1+2\varepsilon)n}$. This concludes the proof of the fact that the required $x_0$ exists.

We take any value that satisfies Equation~\eqref{eq:x0} and further denote it by $x_0$.
Note that, up to a slowly varying multiplier, $x_0$ is of order $n^{\frac{1}{\gamma+1}}$.
Therefore, $O\left(\frac{x_0^\gamma}{n \, L(x_0)}\right) = O\left(n^{-\alpha}\right)$.
From Equations~\eqref{eq:x} and \eqref{eq:x0} it follows that with probability $1 - O\left(\frac{x_0^\gamma}{n \, L(x_0)}\right)$ the number of vertices with degree greater than $x_0$ (i.e., $S_{n,0}(x_0)$) is not larger than $x_0$.
Denote this set of vertices by $A_{x_0}$.
In this case, a clique on $A_{x_0}$ can be constructed.
%If we construct this clique, then we get at least
%$
%(1-\varepsilon)^3 \, n^3 \, x^{- 3\gamma} L^3\left(x\right)
%$
%triangles.

In addition, we want all vertices from the set $A_{x_0}$ to be connected only to each other.
This can be possible, since multiple edges are allowed.
If the sum of degrees in $A_{x_0}$ is odd, then we allow one edge (from the vertex with the smallest degree in $A_{x_0}$) to go outside this set.

We are ready to estimate the global clustering coefficient:
$$
C_1(G_n) = \frac{\text{total value of closed triplets}}{\text{total value of triplets}}\,.
$$
The total value of closed triplets is at least $3 {S_{n,0}(x_0) \choose 3 }$
regardless of the definition of the value of a triplet.
With probability $1 - O\left(\frac{x_0^\gamma}{n \, L(x_0)}\right)$
$$
3 {S_{n,0}(x_0) \choose 3 } \ge \frac 1 2 (1-\varepsilon)^3 \, n^3 \, x_0^{- 3\gamma} L^3\left(x_0\right)\,.
$$
The total value of all triplets includes:
\begin{itemize}
\item The total value of closed triplets on $A_{x_0}$ estimated above,
\item The total value of triplets on the remaining vertices, which is not greater than $\bar S_{n,2}(x_0)$,
\item (optionally) Some unclosed triplets on the vertex with the smallest degree in  $A_{x_0}$, if the sum of degrees in $A_{x_0}$ is odd.
\end{itemize}
Since the smallest degree in the set $A_{x_0}$ is of order $x_0$, we can estimate the last two summands in the total value of triplets by
$$
\bar S_{n,2}(x_0) + O\left(x_0^2\right) \le (1 + \varepsilon) \frac{\gamma}{2 - \gamma}n \, x_0^{2 - \gamma} L\left(x_0\right)\,.
$$
By Lemma~\ref{lem:2}, this holds with probability $1 - O\left(\frac{x_0^\gamma}{n \, L(x_0)}\right)$.

Finally, with probability $1 - O\left(\frac{x_0^\gamma}{n \, L(x_0)}\right)$ we have
\begin{multline*}
C_1(G_n) 
\ge 
\frac{\frac 1 2 (1-\varepsilon)^3 \, n^3 \, x_0^{- 3\gamma} L^3\left(x\right)}{\frac 1 2 (1-\varepsilon)^3 \, n^3 \, x_0^{- 3\gamma} L^3\left(x_0\right) + (1 + \varepsilon) \frac{ \gamma}{2 - \gamma}n \, x_0^{2 - \gamma} L\left(x_0\right)} \\
\ge \frac{\frac 1 2 (1-\varepsilon)^3 \, n^2 \, x_0^{- 2\gamma} L^2\left(x_0\right)}{\frac 1 2 (1-\varepsilon)^3 \, n^2 \, x_0^{- 2\gamma} L^2\left(x_0\right) + (1 + \varepsilon) \frac{ \gamma}{2 - \gamma} \, (1+2\varepsilon)^2\, n^2 \, x_0^{- 2 \gamma} L^2\left(x_0\right)}
\\
= \frac{\frac 1 2 (1-\varepsilon)^3 \, }{\frac 1 2 (1-\varepsilon)^3  + (1 + \varepsilon) \frac{ \gamma}{2 - \gamma} \, (1+2\varepsilon)^2} \ge \frac{2 - \gamma}{2+\gamma} -\delta\,.
\end{multline*}
for sufficiently small $\varepsilon$. Here in the second inequality we used Equation~\ref{eq:x0}.

Recall that the loops are not allowed.
Therefore, it remains to prove that 1) a multi-clique on $A_{x_0}$ can be constructed; 2) a graph on the remaining vertices can be constructed.
Note that a multigraph without loops can always be constructed if the maximum degree is not larger than the sum of the other degrees.

A multi-clique on $A_{x_0}$ can be constructed if
\begin{equation}\label{eq:exist}
\xi_{max} \le S_{n,1}(x_0) - \xi_{max} - x_0^2.
\end{equation}
Here $x_0^2$ is the upper bound for the number of half-edges already involved in the required clique.
From Lemma~\ref{lem:1}, with probability $1 - O\left(n^{-\alpha}\right)$
\begin{equation}\label{eq:exist1}
S_{n,1}(x_0) > (1 - \varepsilon) \frac{\gamma}{\gamma-1}n\,x_0^{1 - \gamma} L\left(x_o\right)\,.
\end{equation}

Fix some $\varepsilon'$ such that $0<\varepsilon' < \frac{1}{\gamma} \left( \frac{\gamma-1}{\gamma+1} - \alpha\right)$.
In this case we have $\alpha < \frac{\gamma-1}{\gamma+1} - \varepsilon'\gamma$, therefore Lemma~\ref{lem:ximax} gives that
\begin{equation}\label{eq:exist2}
\Prob\left(\xi_{max} \le n^{\frac{2}{\gamma+1}-\varepsilon'}\right)
= \Prob\left(\xi_{max} \le n^{\frac{1}{\gamma} + \frac{\gamma-1}{\gamma(\gamma+1)}-\varepsilon'}\right)
= 1 - O\left(n^{-\alpha}\right)\,.
\end{equation}
Now Equation~\eqref{eq:exist} follows immediately from~\eqref{eq:exist1}, \eqref{eq:exist2}, and the fact that $x_0$ is of order $n^\frac{1}{\gamma+1}$.

Similarly, it is easy to show that the graph on the remaining vertices can be constructed:
$$
x_0 \le \bar S_{n,1}(x_0) - x_0
$$
since $\bar S_{n,1}(x_0) = S_{n,1}(0) - S_{n,1}(x_0)$ grows linearly with $n$.
%Note that $\bar S_{n,1}(x_0) = S_{n,1}(0) - S_{n,1}(x_0)$.
%As $S_{n,1}(0)$ grows linearly with $n$, therefore  \bar S_{n,1}(x_0)
%Since  grows linearly with $n$.

\end{proof}

\section{Conclusion}\label{sec:conclusion}

In this paper, we fully analyzed the behavior of the global clustering coefficient in scale-free graphs with an infinite variance of the degree distribution.
We considered both unweighted graphs and graphs with multiple edges.
For the unweighted case, we first obtained the upper bound for the global clustering coefficient.
In particular, we proved that the global clustering coefficient tends to zero with high probability.
We also presented the constructing procedure which allows to reach the obtained upper bound.
The situation turns out to be different for graphs with multiple edges.
In this case, it is possible to construct a sequence of graphs with an asymptotically constant clustering coefficient.

\end{document}